\documentclass[letterpaper, 10 pt, conference]{ieeeconf}  

\IEEEoverridecommandlockouts                              
\usepackage{mathrsfs}
\usepackage{graphicx}
\usepackage{wrapfig,lipsum}
\usepackage{pgfplots}
\usepackage{amsmath,amssymb,amsfonts}  
\usepackage{mathtools}
\usepackage{tikz, tikzscale}
\usepackage{caption}
\usepackage{subcaption}
\usepackage{enumerate}
\usepackage{booktabs}

\newtheorem{lem}{Lemma}
\newtheorem{prop}{Proposition}
\newtheorem{defin}{Definition}
\newtheorem{ass}{Assumption}
\newtheorem{thm}{Theorem}
\newtheorem{remark}{Remark}
\newenvironment{lemma}{\begin{lem}}{\hfill  \end{lem}}

\newenvironment{definition}{\vskip 3pt\begin{defin}}{\vskip 3pt
\end{defin}}

\title{\LARGE \bf
Stability of Open Multi-agent Systems over Dynamic Signed Graphs
}

\author{Pelin Sekercioglu \quad Angela Fontan \quad Dimos V. Dimarogonas
\thanks{This work was supported in part by the Wallenberg AI, Autonomous Systems and Software Program (WASP) funded by the Knut and Alice Wallenberg (KAW) Foundation, the Horizon Europe EIC project SymAware (101070802), the ERC LEAFHOUND Project, the Swedish Research Council (VR), and Digital Futures. P. Sekercioglu, A. Fontan, and D. V. Dimarogonas are with the Division of Decision and Control Systems, KTH Royal Institute of Technology, SE-100 44 Stockholm, Sweden. email:
        {\tt\small \{pelinse,angfon,dimos\}@kth.se}}%
}

\begin{document}

\maketitle
\thispagestyle{empty}
\pagestyle{empty}

\begin{abstract}
This paper addresses the bipartite consensus-control problem in open multi-agent systems containing both cooperative and antagonistic interactions. In these systems, new agents can join and new interactions can be formed over time. Moreover, the types of interactions, cooperative or antagonistic, may change. To model these structural changes, we represent the system as a switched system interconnected over a dynamic signed graph. Using the signed edge-based agreement protocol and constructing strict Lyapunov functions for signed edge-Laplacian matrices with multiple zero eigenvalues, we establish global asymptotic stability of the bipartite consensus control. Numerical simulations validate our theoretical results.
\end{abstract}

\section{Introduction}
Open multi-agent systems (OMAS) are networks where agents and edges can dynamically be added to or removed from the system. They naturally arise in applications such as social networks, where the graph size changes with participant arrivals and departures \cite{hendrickx2016open,hendrickx2017open}, in sensor-based robotic systems, where topology adapts based on inter-agent distances \cite{restrepo2022consensus,dimarogonas2007decentralized}, and in distributed computation \cite{vizuete2022resource}.

Early work on OMAS studied the average consensus problem using a discrete-time gossip algorithm, considering both deterministic \cite{hendrickx2016open} and random \cite{hendrickx2017open} agent arrivals and departures. Additionally, the max-consensus problem was investigated through similar gossip-based interactions \cite{abdelrahim2017max}.
The proportional dynamic consensus problem was later addressed for strongly connected digraphs in \cite{franceschelli2018proportional} and undirected graphs in \cite{franceschelli2020stability}. 
The works \cite{franceschelli2018proportional, franceschelli2020stability} also introduce novel stability definitions for OMAS, based on a new ``open'' distance, which extends the notion of Euclidean distance to vectors belonging to Euclidean spaces of different dimensions.
A different stability framework for the consensus problem in OMAS was proposed in \cite{xue2022stability}. The OMAS was represented as a switched system, where dynamical systems of different dimensions and over different graph topologies (the modes) are active at different time instants based on a switching signal. This framework incorporates a transition-dependent average-dwell-time condition and allows for disconnected graphs. Following \cite{xue2022stability}, \cite{restrepo2022consensus} tackled the problem of consensus control for first-order systems over undirected graphs, ensuring all-to-all collision avoidance. Other switched systems-based approaches include \cite{restrepo2024distributed}, which introduces a gradient-based controller to maintain biconnectivity despite agent removal.

In all the references above, agents are assumed to cooperate with each other. However, many real-world scenarios involve antagonistic behaviors, such as herding control \cite{sebastian2022adaptive}, aerospace applications with attitude control of multiple rigid bodies \cite{mathavaraj2023rigid}, and social networks where agents compete \cite{altafini2012_6329411,Fontan2021Role,Fontan2021Signed} or spread disinformation \cite{csekerciouglu2024distributed}. A common framework for modeling both cooperation and antagonism in dynamical systems is that of \textit{signed graphs} \cite{altafini2012_6329411,Shi2019Dynamics}. In such networks, the emerging behavior is typically \textit{bipartite consensus}, where agents converge to the same value in modulus but not in sign, provided the underlying graph is \textit{structurally balanced}, meaning that the agents can be divided into two disjoint subsets where intra-group agents cooperate and inter-group agents compete \cite{altafini2012_6329411}. Bipartite consensus control has been studied, \textit{e.g.,} for single or double-integrators \cite{altafini2012_6329411, yang2019bipartite} and linear high-order dynamics \cite{valcher}.

In this paper, we study the bipartite consensus of OMAS over undirected signed graphs. To the best of our knowledge, this is the first attempt to address this problem. We consider systems where new nodes and edges can be added, and interconnections may switch between cooperation and antagonism. This work reflects real-world scenarios such as social networks, where relationships evolve and shifts in opinion or alliances occur \cite{antal2006social}, and robotic networks, where dynamic topology changes due to environmental factors or mission updates influence collaboration and competition \cite{6213085}.

We propose an approach for studying OMAS interconnected over dynamic signed graphs, modeling them as switched systems based on the method proposed in \cite{xue2022stability} for cooperative networks. Focusing on first-order systems, we establish global asymptotic stability of the bipartite consensus set via Lyapunov’s direct method—an important step toward extending our framework to nonlinear systems. We reformulate the bipartite consensus problem as a stability problem for synchronization errors, analyzed in signed edge-based coordinates. Building on \cite{ECC20_EP}, which provides a Lyapunov characterization for matrices with a single zero eigenvalue, we extend this result to multiple zero eigenvalues and specifically address signed edge-Laplacians matrices (hereafter: edge Laplacians), unlike \cite{csekerciouglu2024distributed}. Our contributions include the construction of strict Lyapunov functions and analysis of dynamic signed graphs, where the addition of nodes and signed edges directly impacts the system's convergence behavior.

\textit{Notation:} $\lvert \cdot \rvert$ denotes the absolute value for scalars, the Euclidean norm for vectors, and the spectral norm for matrices. $\mbox{card}(\cdot)$ indicates the cardinality of a set. $\mbox{diag}(z)$ denotes a diagonal matrix whose diagonal elements are the entries of the vector $z$. $\mathcal{N}(A)$ denotes the null space of matrix $A$. $\mathbb{R}$ is the set of real numbers and $\mathbb{R}_{\geq 0}$ the nonnegative orthant. $A>0$ ($A\ge 0$) indicates that $A$ is a positive definite (positive semidefinite) matrix.

\section{Preliminaries on signed graphs and edge-based formulation}\label{section3}
Let $\mathcal{G}_s = (\mathcal{V}, \mathcal{E})$ be a signed graph, where $\mathcal{V} = \{ \nu_1, \nu_2, \dots, \nu_N \}$ is the set of nodes and $\mathcal{E} \subseteq \mathcal{V} \times \mathcal{V}$ is the set of $M$ edges. Each edge in $\mathcal{E}$ has a sign (either positive or negative). The graph is \textit{undirected} if the information flow between agents is bidirectional, meaning $( \nu_i, \nu_j ) = ( \nu_j, \nu_i ) \in \mathcal{E}$. Otherwise, the graph is \textit{directed} and referred to as a \textit{digraph}. Self-loops are not considered. 
The adjacency matrix of \( \mathcal{G}_s \) is \( A := [a_{ij}] \in \mathbb{R}^{N \times N} \), where \( a_{ij} \neq 0 \) if and only if \( (\nu_j, \nu_i) \in \mathcal{E} \). Each edge $(\nu_j, \nu_i)$ has either a positive sign (\( a_{ij} > 0 \)), representing a cooperative relationship, or a negative sign (\( a_{ij} < 0 \)), representing an antagonistic relationship.
A spanning tree of a graph is a subgraph that includes all the vertices of the original graph and is a tree, meaning it is connected and has no cycles. A signed graph is said to be \textit{structurally balanced} (SB) if it can be split into two disjoint sets of vertices $\mathcal{V}_{1}$ and $\mathcal{V}_{2}$, where $\mathcal{V}_{1} \cup \mathcal{V}_{2} = \mathcal{V}$ and $\mathcal{V}_{1} \cap \mathcal{V}_{2} = \emptyset$ such that for every $\nu_{i}, \nu_{j} \in \mathcal{V}_{p}, p \in \{ 1,\, 2 \}$, we have $a_{ij} \geq 0$, while for every $\nu_{i} \in \mathcal{V}_{p}, \nu_{j} \in \mathcal{V}_{q}$, with $p,q \in \{ 1, \, 2 \}, p \neq q$, we have $a_{ij} \leq 0$. It is {\it structurally unbalanced} (SUB), otherwise.

The elements of the signed Laplacian matrix $L_{s}=[{\ell _{s_{ij}}}] \in \mathbb{R}^{N \times N}$ associated with $\mathcal{G}_s$ are defined as \cite{altafini2012_6329411,Shi2019Dynamics} 
\begin{equation}\label{Ls}
	{\ell _{s_{ij}}} = \left\{ {\begin{array}{*{20}{c}}
			{\sum\limits_{k \leq {N}} {{ \lvert a_{ik} \rvert }} }&{i = j}\\
			{ - {a_{ij}}}&{i \ne j}.
	\end{array}} \right.
\end{equation}
In the following definition, we introduce the incidence matrix of a signed graph \cite{du2019further,sekercioglu2024control}.
\begin{definition}
Consider a signed graph $\mathcal{G}_s$ containing $N$ nodes and $M$ edges. The signed
incidence matrix $E_s \in \mathbb{R}^{N \times N}$ of $\mathcal{G}_s$ is defined as
    \begin{equation}\label{434}
	{[E_{s}]_{ik}} :=  \left\{ \begin{matrix*}[l]
		{+1,}&{\text{if}\ \nu_i\ \text{is the initial node of}\ \varepsilon_k ;}\\
		{-1,}&{\text{if } \nu_i,\nu_j\ \text{are cooperative and}}\\ 
{}&{\nu_i \text{ is the terminal node of}\ \varepsilon_k ;}\\
{+1,}&{\text{if } \nu_i,\nu_j\ \text{are competitive and}}\\
{}&{\nu_i \text{ is the terminal node of}\ \varepsilon_k ;}\\
{0,}&{\text{otherwise}},
\end{matrix*}\right.
\end{equation}
where $\varepsilon_k$ is the arbitrarily oriented edge interconnecting nodes $\nu_i$ and $\nu_j$,\ $k \leq M,\ i,j \leq N$.
\end{definition}

The Laplacian matrix $L_{s} \in \mathbb{R}^{N \times N}$ and the edge Laplacian matrix $L_{e_{s}} \in \mathbb{R}^{M \times M}$ of an undirected signed graph $\mathcal{G}_s$ can be expressed as
\begin{align}\label{IJC_453}
L_{s} = E_{s}E_{s}^\top,\quad L_{e_{s}} = E_{s}^\top E_{s}. 
\end{align}
The signed Laplacian matrix of an undirected and connected signed graph is symmetric, which means that all its eigenvalues are real. Moreover, for such a graph, the following properties hold regarding the positive (semi-)definiteness of its Laplacians and the dimension of their null space.

\begin{lemma} (\cite[Lemma 7]{du2019further}) \label{N_E_undirected}
    For an undirected SB graph containing a spanning tree, $\mathcal{N}(L_{e_s}) = \mathcal{N}(E_s)$ holds.
\end{lemma}
\begin{lemma}\label{lemma1}
For an undirected and SB graph containing a spanning tree, the following holds. 
\begin{enumerate}[{\normalfont (i)}]
	\item $L_{s}$ has a simple eigenvalue 0 and all the other eigenvalues are positive. 
	\item $\mbox{rank}(L_{s}) = \mbox{rank}(L_{e_{s}}) = \mbox{rank}(E_s) = N-1$. 
\end{enumerate}
\end{lemma}
\begin{proof}
	Item (i) follows directly from \cite[Lemma 1]{altafini2012_6329411}. For item (ii), since $L_{s}$ is the Laplacian of an SB graph, it follows that $\mbox{rank}(L_{s}) = N-1$. Then, the structural balance property allows a node-based gauge transformation on $L_{s}$, and an edge-based gauge transformation on $L_{e_{s}}$. That is, by \cite{altafini2012_6329411}, there exists a diagonal matrix $D \in \mathfrak{D}$, where $\mathfrak{D} = \{ D = \mbox{diag}(d),\ d = [d_{1}, d_{2}, \dots, d_{{N}}],\ i \leq N \}$, where $d_{i} \in \{ -1, 1\}$, such that $D L_{s} D$ corresponds to the Laplacian of a graph with only cooperative interactions (unsigned graph). Similarly, an edge gauge transformation \cite[Lemma 1]{du2019further} on the unsigned edge Laplacian $L_{e}$ gives $L_e = D_e L_{e_{s}} D_{e}$, where $D_{e} = \mbox{diag}(d_{e})$, with ${d_{e} = [d_{e_{1}},\ \dots,\ d_{e_{{M}}}],}$ with $d_{e_{i}} = 1$, $i \leq M$ if $\nu_{i} \in \mathcal{V}_{1}$ and $d_{e_{i}} = -1$ if $\nu_{i} \in \mathcal{V}_{2}$ with $\nu_{i}$ being the initial node of the edge. Since gauge transformations are similarity transformations, they preserve the spectrum, so Item (ii) follows from \cite[Lemma 2]{zeng2014nonlinear}, \cite[0.4.6]{Horn}, and \cite[Lemma 1]{du2019further}. 
\end{proof}

\begin{lemma}\label{lemma2}
	For an undirected and SUB graph containing a spanning tree, the following holds. 
	\begin{enumerate}[{\normalfont (i)}]
		\item $L_{s}$ has only positive eigenvalues.
		\item $\mbox{rank}(L_{s}) = \mbox{rank}(L_{e_{s}}) = \mbox{rank}(E_s) = N$. 
	\end{enumerate}
\end{lemma}

\begin{proof}
	Item (i) follows directly from \cite[Corollary 2]{altafini2012_6329411}. Then, since $L_{s}$ is the signed Laplacian matrix of a SUB graph, it only has positive eigenvalues, which implies $\mbox{rank}(L_{s}) = N$. Suppose $\lambda \neq 0$ is an eigenvalue of $L_{s}$, which is associated with a non-zero eigenvector $v_{r}$, such that $L_{s} v_{r} = E_{s}E_{s}^\top v_{r} = \lambda v_{r} \neq 0.$ It is clear that $E_{s}^\top v_{r} \neq 0$. Let $\bar v_{r} = E_{s}^\top v_{r}$. Then, by left-multiplying $E_{s}^\top$ on both sides, we obtain $E_{s}^\top E_{s}E_{s}^\top v_{r} = E_{s}^\top \lambda v_{r}$. By replacing \eqref{IJC_453} in the latter, we obtain $L_{e_{s}} \bar v_{r} = \lambda \bar v_{r},$ which means $\lambda$ is also an eigenvalue of $L_{e_{s}}$. Consequently, the nonzero eigenvalues of $L_{s}$ and $L_{e_{s}}$ are equal to each other, and since they have the same nonzero eigenvalues, their rank is also equal. Furthermore, Item (ii) follows with \cite[Lemma 2]{du2019further}. 
\end{proof}

\section{Model and Problem Formulation}\label{section2}
We consider a signed network of agents with dynamic cooperative and antagonistic interactions, where agents can join, edges may be added, and interactions can change over time. These changes in the multi-agent system can be modeled using a switched system representation \cite{xue2022stability}.

Let $\sigma : \mathbb{R}_{\geq 0} \to \mathcal{P}$ be the switching signal associated with topology changes, where $\mathcal{P} := \{1, 2, \dots, s\}$ represents the set of $s$ possible switching modes. Each mode of the system is denoted by $\phi \in \mathcal{P}$, with $\phi = \sigma(\tau)$ for $\tau \in [t_l, t_{l+1})$, where $t_l$ and $t_{l+1}$ are consecutive switching instants. A switching instant $t_l$ is defined as the time at which a new edge is created between agents, a new agent enters the system, or the sign of an interconnection changes. $T_{\phi}$ denotes the total duration during which mode $\phi$ is active. We pose the following assumption.
\begin{ass}\label{ass0}
The total number of possible switching modes is finite, that is, $\mbox{card}(\mathcal{P}) < \infty$.
\end{ass}

Next, we introduce the concept of transition-dependent average dwell time for the switching signal $\sigma(t)$, which ensures that the switching signal meets the required conditions for stability--- See Theorems \ref{thmxue} and \ref{prop:result1} in Section \ref{section5}. 

\begin{definition} (\cite[Definition 2]{xue2022stability}) On a given interval $[t_0, t_f)$, with $t_f > t_0 \geq 0$, consider any two consecutive modes $\hat{\phi}, \phi \in \mathcal{P}$, where $\hat{\phi}$ precedes $\phi$. Let $N_{\hat{\phi},\phi}(t_0, t_f)$ denote the total number of switchings from mode $\hat{\phi}$ to mode $\phi$, and let $T_{\phi}$ denote the total active duration of mode $\phi$. The constant $\tau_{\hat{\phi},\phi} > 0$ satisfying
\begin{equation}\label{avr_dwell_time}
    N_{\hat{\phi},\phi}(t_0, t_f) \leq \hat{N}_{\hat{\phi},\phi} + \frac{T_{\phi}(t_0, t_f)}{\tau_{\hat{\phi},\phi}},
\end{equation}
for any given scalar $\hat{N}_{\hat{\phi},\phi} \geq 0$, is called the \textit{transition-dependent average dwell time} of the switching signal $\sigma(t)$.
\end{definition}

At each mode $\phi \in \mathcal{P}$, consider the OMAS composed of $N_{\phi}$ dynamical systems modeled by
\begin{align}\label{FO}
	\dot x_i = u_i,\quad i \in \{1, 2, \dots, N_{\phi} \},
\end{align}
where $x_i \in \mathbb{R}$ is the state of the $i$th agent, and $u_i \in \mathbb{R}$ is the control input. 
The agents interact on a dynamic signed graph $\mathcal{G}_{s _{\phi}}$ that contains $N_{\phi}$ nodes and $M_{\phi}$ edges. The system \eqref{FO} is interconnected with the control law
\begin{align}\label{CL}
	u_i = -k_1 \sum_{j=1}^{N_{\phi}} \lvert a_{\phi_{ij}} \rvert \left[ x_i - \mbox{sign}(a_{\phi_{ij}} )x_j \right], 
\end{align}
where $k_1>0$, and $A_{\phi} = [a_{\phi_{ij}}]$ is the adjacency matrix at mode $\phi \in \mathcal{P}$, with $a_{\phi_{ij}} \in \{ 0, \pm 1 \}$ representing the adjacency weight between nodes $\nu_j$ and $\nu_i$. 
Moreover, we pose the following assumption on the connectivity of the initial signed graph.

\begin{ass}\label{ass1}
    The initial signed graph contains a spanning tree.
\end{ass}
\begin{remark} 
We consider a dynamic signed graph with node additions to model systems where agents progressively join while maintaining existing interactions. To preserve graph connectivity, which is crucial for sustaining the collective behavior of the network, we allow the removal of the most recently added nodes. This captures scenarios where new agents influence cooperation and antagonistic behaviors, which is particularly relevant in applications like multi-robot coordination and opinion dynamics.
\end{remark}

The possible control objectives for \eqref{FO} interconnected over a signed graph depend on the \textit{structural balance} property. 
Under Assumption \ref{ass1}, the achievable control objective for \eqref{FO} interconnected over an SB graph is to ensure agents achieve \textit{bipartite consensus}, where agents converge to the same value in modulus but not in sign, that is,
\begin{align}\label{obj_BC}
	\lim_{t \to \infty} \left[ x_{i}(t) - \mbox{sign}(a_{\phi_{ij}})x_{j}(t) \right] = 0, \quad \forall i,j \leq N_{\phi}.
\end{align}

If the considered graph is SUB, the achievable control objective under Assumption \ref{ass1} for \eqref{FO} is \textit{trivial consensus}, where all agents converge to zero, that is,
\begin{align}\label{obj_C}
	\lim_{t \to \infty} x_{i}(t) = 0, \quad \forall i \leq N_{\phi}.
\end{align} 
Under Assumption~\ref{ass0}, in this paper we show that, depending on whether the signed graph associated with the last switching mode is SB or SUB, synchronization in the OMAS leads to bipartite consensus \eqref{obj_BC} or trivial consensus \eqref{obj_C}.


The control objectives for \eqref{FO} can also be expressed in terms of the synchronization errors, defined as
\begin{equation}\label{def_e}
    e_k = x_{i} - \mbox{sign}(a_{\phi_{ij}})x_{j},\quad k =(\nu_j,\nu_i) \in \mathcal{E}_{\phi}
\end{equation}
where $k \leq M_{\phi}$ denotes the index of the interconnection between the $j$th and $i$th agents, and is equivalent to
\begin{align}\label{obj}
	\lim_{t \to \infty} e_{k}(t) = 0, \quad \forall k \leq M_{\phi}.
\end{align}


The definition of the synchronization errors \eqref{def_e} corresponds to the \textit{(signed) edge-based formulation} of a network. Before introducing our main results in Section~\ref{section5}, we derive in Section~\ref{section4} some properties of the signed edge Laplacians that are useful to establish our main results, using notions on signed graphs and the signed edge-based formulation from \cite{altafini2012_6329411} and \cite{du2019further}, respectively.

\section{Lyapunov Equation for Signed Edge Laplacians}\label{section4}
To establish bipartite consensus of open multi-agent systems over signed graphs (Section~\ref{section5}), we will show how to construct strict Lyapunov functions in the space of $e$, \textit{i.e.,} expressed in terms of the synchronization errors defined in \eqref{def_e}.
The analysis is based on the following properties of the signed edge Laplacians. The following statements are original contributions of this paper and extend \cite[Proposition~1]{ECC20_EP} to the case of signed Laplacians containing multiple zero eigenvalues (Theorem~\ref{lemma:lyap_eq_spanningtree}), and \cite[Proposition 1]{csekerciouglu2024distributed} to the case of signed edge Laplacians corresponding to undirected graphs with multiple zero eigenvalues (Theorem~\ref{th:lyap_eq_Le}). 

For a spanning tree graph, we have the following.

\begin{thm} \label{lemma:lyap_eq_spanningtree} Let $\mathcal{G}_{s}$ be a signed graph containing $N$ agents interconnected by $M$ edges, and let $L_{e_{s}}$ be the associated edge Laplacian. If the graph $\mathcal{G}_{s}$ is a spanning tree, then for any $Q \in \mathbb{R}^{(N-1) \times (N-1)}, Q=Q^{\top}> 0$, there exists a matrix $P \in \mathbb{R}^{(N-1) \times (N-1)}$, $P=P^{\top}>0$ such that 
\begin{equation}\label{Lyap-eq}
    PL_{e_{s}} + L_{e_{s}}^{\top}P = Q.
\end{equation}
\end{thm} 

\begin{proof} By assumption, the graph $\mathcal{G}_{s}$ is a spanning tree. Then, $\mathcal{G}_{s}$ is SB, it consists of $N-1$ edges, and $L_{e_{s}} \in \mathbb{R}^{(N-1)\times (N-1)}$. Since the non-zero eigenvalues of $L_{e_{s}}$ and $L_{{s}}$ coincide (Lemma \ref{lemma1}), the rank of $L_{e_{s}}$ is $N-1$ and $L_{e_{s}}$ has positive eigenvalues, which implies that $-L_{e_{s}}$ is Hurwitz. Then, given any symmetric positive definite matrix $Q \in \mathbb{R}^{(N-1) \times (N-1)}$, there exists a symmetric positive definite matrix $P \in \mathbb{R}^{(N-1) \times (N-1)}$ such that $PL_{e_{s}} + L_{e_{s}}^{\top}P = Q$ holds.\end{proof}

For an undirected graph containing a spanning tree, we have the following.

\begin{thm} \label{th:lyap_eq_Le} Let $\mathcal{G}_s$ be a signed undirected graph with $N$ nodes and $M$ edges, and let $L_{e_{s}}$ be the associated edge Laplacian, which contains $\xi$ zero eigenvalues. Then the following are equivalent:
	\renewcommand{\theenumi}{(\roman{enumi}}
	\begin{enumerate}
		\item $\mathcal{G}_s$ contains a spanning tree,
		\item for any $Q \in \mathbb{R}^{M \times M}, Q=Q^{\top}> 0$ and for any $\{ \alpha_{1}, \alpha_{2}, \dots, \alpha_{{\xi}}\}$ with $\alpha_{i} >0$, there exists a matrix $P(\alpha_{i}) \in \mathbb{R}^{M \times M}$, $P=P^{\top}>0$ such that
		\begin{equation} \label{eq:lyapunov_eq_multiple_leaders_Le}
			PL_{e_{s}} = \frac{1}{2} \left[ Q - \sum_{i=1}^{\xi} \alpha_{i} (Pv_{r_{i}}v_{l_{i}}^{\top} + v_{l_{i}}v_{r_{i}}^{\top}P)\right],
		\end{equation}
		where $v_{r_{i}}, v_{l_{i}} \in \mathbb{R}^N$ are, respectively, the right and left eigenvectors of $L_{e_{s}}$ associated with the $i$th 0 eigenvalue.
	\end{enumerate} 
Moreover, if the signed graph is SB, $\xi = M-N+1$, and $\xi = M-N$ otherwise.
\end{thm} 

\begin{proof}
	(i) $\Rightarrow$ (ii): By assumption, the graph $\mathcal{G}_{s}$ contains a spanning tree. Then, it follows that $L_s$ has either a unique zero eigenvalue and $N-1$ positive eigenvalues if SB: $0 = \lambda_1 < \lambda_2 \leq \dots \leq \lambda_{N}$, or $N$ positive eigenvalues if SUB: $0 < \lambda_1 \leq \lambda_2 \leq \dots \leq \lambda_{N}$. Since the non-zero eigenvalues of $L_{s}$ and $L_{e_{s}}$ coincide (Lemmata \ref{lemma1} and \ref{lemma2}), $L_{e_{s}}$ has $\xi$ zero eigenvalues: $0 = \lambda_1 = \dots = \lambda_{\xi} < \lambda_{\xi+1} \leq \dots \leq \lambda_{M}$, where $\xi = M - N + 1$ if the graph is SB and $\xi = M - N$, otherwise. We can write the Jordan decomposition of $L_{e_{s}}$ as $L_{e_{s}} = U \Lambda U^\top = \sum_{i=1}^{\xi} \lambda_{i}(L_{e_{s}})v_{r_{i}}v_{l_{i}}^{\top} + U_{1} \Lambda_{1} U_{1}^\top$ where $ \Lambda_{1} \in \mathbb{R}^{(M - \xi) \times (M - \xi)}$ contains the positive eigenvalues of $L_{e_{s}}$, and $U = \left[ v_{r_{1}}, \dots, v_{r_{{\xi}}}, U_{1} \right] \in \mathbb{R}^{M \times M}$, $U^\top = \left[ v_{l_{1}}, \dots, v_{l_{{\xi}}}, U_{1}^\top \right]^\top \in \mathbb{R}^{M \times M}$ consist of the right and left eigenvectors of $L_{e_{s}}$, along with $U_{1}$, which contains the remaining $M-\xi$ columns. For any $\{ \alpha_{1}, \alpha_{2}, \dots, \alpha_{{\xi}}\}$ with $\alpha_{i} >0$  define $R = L_{e_{s}} + \sum_{i=1}^{\xi} \alpha_{i}v_{r_{i}}v_{l_{i}}^{\top}$. From this decomposition and the fact that $\Lambda_{1}$ contains only positive eigenvalues, $ \lambda_{j}(R) > 0$ for all $j \leq M$. $-R$ is Hurwitz, therefore for any ${Q = Q^{\top}>0}$ and $\alpha_{i} >0, i \leq \xi$, there exists $P = P^{\top}>0$ such that $-PR-R^{\top}P = -Q$,
	\begin{align*}
		&-P [ L_{e_{s}} + \sum_{i=1}^{\xi} \alpha_{i}v_{r_{i}}v_{l_{i}}^{\top} ] - [L_{e_{s}} + \sum_{i=1}^{\xi} \alpha_{i}v_{r_{i}}v_{l_{i}}^{\top} ]^{\top}P = -Q,\\
		&2PL_{e_{s}} = Q - \sum_{i=1}^{\xi} \alpha_{i} \left(Pv_{r_{i}}v_{l_{i}}^{\top} + v_{l_{i}}v_{r_{i}}^{\top}P\right).
	\end{align*}
	(ii) $\Rightarrow$ (i): Let statement (ii) hold and assume that the signed edge Laplacian has $\xi+1$ zero eigenvalues and the rest of its $M-\xi-1$ eigenvalues are positive. In view of Lemmata \ref{lemma1} and \ref{lemma2}, neither the property that $L_{e_{s}}$ has $\xi$ positive eigenvalues ---with $\xi = M - N + 1$ if the graph is SB, and $\xi = M - N$ if SUB--- nor the fact that the graph contains a spanning tree hold. Now, the Jordan decomposition of $L_{e_{s}}$ has the form $L_{e_{s}} = U \Lambda U^\top = \sum_{i=1}^{\xi+1} \lambda_{i}(L_{e_{s}})v_{r_{i}}v_{l_{i}}^{\top} + U_{1} \Lambda_{1} U_{1}^\top$ with $U = \left[ v_{r_{1}}, \dots, v_{r_{{\xi+1}}}, U_{1} \right] \in \mathbb{R}^{M \times M}$. Next, to analyze the definiteness of the matrix, let us consider $R(\alpha_{i}) =L_{e_{s}} + \sum_{i=1}^{\xi} \alpha_{i}v_{r_{i}}v_{l_{i}}^{\top}$ which admits the Jordan decomposition $R:= U \Lambda_{R} U^\top$, where {\small
	$$
	\Lambda_{R} := \begin{bmatrix} \alpha_{1} &  &  & & \\[-5pt]  &\hspace{-1.3ex} \ddots & & & \\[-5pt]  &  &\hspace{-1ex} \alpha_{{\xi}} & & \\[1pt]   &  & &\hspace{-1em} 0 & \\[-1pt]  & &  & &\hspace{-1ex}  \Lambda_{1} \end{bmatrix}.
	$$}
	Since $R$ has one zero eigenvalue, it cannot be positive definite, as positive definiteness requires that all eigenvalues be strictly positive. Then, there exists a matrix $Q = Q^{\top}$ for which no symmetric matrix $P=P^{\top}$ satisfies $-PR-R^{\top}P = -Q$. This contradicts statement (ii), completing the proof. \end{proof}

\section{Bipartite Consensus on OMAS}\label{section5}
In this section, we present our main results. We consider the bipartite consensus-control problem of OMAS modeled by \eqref{FO} over undirected signed graphs, and establish the practical asymptotic stability of the synchronization errors. Before delving into the details, we first recall the following crucial theorem and definition from \cite{xue2022stability}, which serve as the foundation for our analysis and the main results. Next, we derive the synchronization errors using the edge-Laplacian notation introduced in Section \ref{section3}, and reformulate the control problem accordingly. 

Consider the following switched system 
\begin{align}\label{app1}
    \dot{x}_{\phi}(t) = f_{\phi}(x_{\phi}(t)),
\end{align}
where $x_{\phi} \in \mathbb{R}^{N_\phi}$, indicating that the system's dimension may change with each switch. For this system, we have the following.
\begin{definition}
The system \eqref{app1} is said to be globally uniformly practically stable, if there exist a class $\mathcal{KL}$ function $\beta$ and a scalar $\epsilon \geq 0$ such that for any initial state $x_{\phi}(t_0)$ and admissible $\phi$,
\[   
\|x_{\phi}(t)\| \leq \beta(\|x_{\phi}(t_0)\|, t - t_0) + \epsilon, \quad \forall t \geq t_0,
\] 
where $\epsilon$ is called the ultimate bound of $x_{\phi}(t)$ as $t \to +\infty$. Particularly, if one has $\epsilon = 0$, then \eqref{app1} is said to be globally uniformly asymptotically stable.
\end{definition}
\begin{thm} (\cite{xue2022stability})\label{thmxue}
    Consider the system \eqref{app1} with the switching signal $\sigma(t)$ on $[t_0,t_f],$ $0\leq t_0 < t_f < + \infty$. Assume that, for any two consecutive modes $\phi$, $\hat \phi \in \mathcal{P}$, where $\hat \phi$ precedes $\phi$, there exist class $\mathcal{K}_{\infty}$ functions $\underline{\kappa}$, $\overline{\kappa}$, constants $\gamma_{\phi}>0,$ $\Omega_{\phi, \hat \phi}>0$, $\Theta \geq 0$, and a non-negative function $V_{\phi}(t, x_{\phi}(t)) : \mathbb{R}_{\geq 0} \times \mathbb{R}^{N_{\phi}} \to \mathbb{R}_{\geq 0}$, such that,
    \begin{align*} 
        &\underline{\kappa} (\lvert x_{\phi} \rvert) \leq V_{\phi}(t, x_{\phi}(t)) \leq \overline{\kappa} (\lvert x_{\phi} \rvert) \\
        &\dot V_{\phi}(t, x_{\phi}(t)) \leq -\gamma_{\phi} V_{\phi}(t, x_{\phi}(t))\\
        &V_{\phi}(t_k^+, x_{\phi}(t_k^+)) \leq \Omega_{\phi, \hat \phi} V_{\hat \phi}(t_k^-, x_{\hat \phi}(t_k^-)) + \Theta,
    \end{align*}
     $\forall t \in [t_0, t_f].$ Moreover, asssume that $\sigma(t)$ satisfies
    \begin{align*}
        \tau_{\phi,\hat{\phi}} \geq \frac{\ln(\Omega_{\phi,\hat{\phi}})}{\gamma_{\phi}},
    \end{align*}
    where $\tau_{\phi,\hat{\phi}}$ is defined in \eqref{avr_dwell_time}. Then, \eqref{app1} is globally uniformly practically stable. In particular, if $\Theta = 0$, then \eqref{app1} is globally uniformly asymptotically stable.
\end{thm} 

We now proceed to derive the synchronization errors using the edge-Laplacian notation from Section \ref{section3}. At each mode $\phi \in \mathcal{P}$, using the definition of incidence matrix \eqref{434}, we may express the synchronization errors \eqref{def_e} and  control law \eqref{CL} in vector form as
\begin{align}
	e_{\phi} &= E_{s_{\phi}}^\top x,\label{err_vect}\\
    u_{\phi} &= -k_1  E_{s_{\phi}} {e}_{{\phi}}, \label{CL_vect_edge}
\end{align}
where $E_{s_{\phi}}$ is the signed incidence matrix corresponding to the graph at mode $\phi$. 
Differentiating the edge states, we obtain for mode $\phi$
\begin{align}
	\dot{e}_{{\phi}} = -k_1 E_{s_{\phi}}^\top E_{s_{\phi}} E_{s_{\phi}}^\top x_{\phi}= -k_1 L_{e_{s_{\phi}}} e_{{\phi}}
\end{align}
from the definition of the edge Laplacian \eqref{IJC_453}. Following the approach proposed in \cite{xue2022stability}, we can capture the edge state transition between modes at each switching instant \( t_l \), as follows: 
\begin{align}\label{edge_switch1}
	e_{{\phi}}(t_l^+) = \Xi_{\phi, \hat \phi} e_{{\hat \phi}}(t_l^-) + \Phi_l,
\end{align}
where $\phi = \sigma(\tau) \in \mathcal{P}$ represents the active mode after switching, with $\tau \in [t_l, t_{l+1})$, while $\hat \phi = \sigma(\hat \tau) \in \mathcal{P}$ denotes the previous mode, where $\hat \tau \in [t_{l-1}, t_l)$. The matrix $\Xi_{\phi, \hat \phi} \in \mathbb{B}^{M_{\phi} \times M_{\hat \phi}}$ is a matrix with elements in $\{ 0, 1\}$ that determines how the edge state \( e_{{\hat{\phi}}}(t_l^-) \) transitions between modes. Here, \( e_{{\hat{\phi}}}(t_l^-) \) represents the edge states just before the switching instant, while \( e_{{\phi}}(t_l^+) \) represents the edge states just after the switching instant. $\Phi_l \in \mathbb{R}^{M_{\phi}}$ is a real-valued vector that captures instantaneous changes in the edge state $e_{{\phi}}$ at the switching moment $t_l$. These changes may arise from the addition or removal of a node, which leads to the creation or deletion of an edge, or from a sign change in the interaction. Then, the closed-loop system for the error dynamics is given as
\begin{subequations}
\begin{align}
  &\dot{e}_{{\phi}}(t) = -k_1 L_{e_{s_{\phi}}} e_{{\phi}}(t),\quad t \in [t_l, t_{l+1}),\label{dyn_edge}\\
  &e_{{\phi}}(t_l^+) = \Xi_{\phi, \hat \phi} e_{{\hat \phi}}(t_l^-) + \Phi_l,\quad t = t_l. \label{edge_switch}
\end{align}
\end{subequations}

We are now ready to present our two main results, namely practical asymptotic stability (Theorem \ref{prop:result1}) and asymptotic stability (Theorem \ref{prop:result2}) of the origin of the error dynamics \eqref{dyn_edge}. Let the control objective for \eqref{dyn_edge} be to ensure that the origin is globally asymptotically practically stable, that is,
\begin{align}\label{obj_e}
	\lim_{t \to \infty} \lvert e_{{\phi}}(t) \rvert \leq \varepsilon.
\end{align}
Then, we have the following.
\begin{thm}\label{prop:result1}
	Consider the OMAS \eqref{FO}, under Assumption \ref{ass1}, in closed loop with the switching control law \eqref{CL_vect_edge}. Let $\phi, \hat \phi \in \mathcal{P}$ be two consecutive modes. Then, if the switching signal $\sigma$ admits an average dwell time satisfying
	\begin{equation}\label{cond}
		\tau_{\phi,\hat{\phi}} \geq \frac{\ln(\Omega_{\phi,\hat{\phi}})}{\gamma_{\phi}}
	\end{equation} 
	where $\Omega_{\phi,\hat{\phi}}$ and $\gamma_{\phi}$ are positive constants defined further below, the origin of the closed-loop system \eqref{dyn_edge}--\eqref{edge_switch} is asymptotically practically stable for all initial conditions.
\end{thm}

\begin{proof} For each mode $\phi$, let $Q_{\phi} = Q_{\phi}^\top > 0$ and $\alpha_{\phi}$ be arbitrarily fixed. By Assumption \ref{ass1} and Theorem \ref{th:lyap_eq_Le}, there exists a symmetric positive definite matrix $P_{\phi} = P_{\phi}^\top > 0$ such that \eqref{eq:lyapunov_eq_multiple_leaders_Le} holds. In the case that the considered graph is a spanning tree,  there exists a $P_{\phi} = P_{\phi}^\top > 0$ such that \eqref{Lyap-eq} holds. Then, consider the Lyapunov function candidate
	\begin{align}\label{V}
		V_{\phi}(e_{{\phi}}) := \frac{1}{2} e_{{\phi}}^\top P_{\phi} e_{{\phi}}.
	\end{align}
	Then, for all $\tau \in [t_l, t_{l+1})$, the derivative of \eqref{V} along the trajectories of \eqref{dyn_edge} yields
	\begin{align*}
		\dot V_{\phi}(e_{{\phi}}(\tau)) = - k_1 e_{{\phi}}(\tau)^\top P_{\phi} L_{e_{s_{\phi}}} e_{{\phi}}(\tau).
	\end{align*}
	If the underlying signed graph is a spanning tree, using \eqref{Lyap-eq} with $Q_{\phi} = I_{M_{\phi} \times M_{\phi}}$, we obtain
	\begin{align*}
		\dot V_{\phi}(e_{{\phi}}(\tau)) &= - \frac{1}{2} k_1 e_{{\phi}}(\tau)^\top e_{{\phi}}(\tau).
	\end{align*}
	If the underlying signed graph contains a spanning tree, using \eqref{eq:lyapunov_eq_multiple_leaders_Le} with $Q_{\phi} = I_{M_{\phi} \times M_{\phi}}$, we obtain
	\begin{align*}
		&\dot V_{\phi}(e_{{\phi}}(\tau)) = - \frac{1}{2} k_1 e_{{\phi}}(\tau)^\top e_{{\phi}}(\tau) \\
		 &- \frac{1}{2} e_{{\phi}}(\tau)^\top \left( \sum_{i=1}^{\xi_{\phi}} \alpha_{\phi_i} (P_{\phi}v_{r_{\phi_i}}v_{l_{\phi_i}}^{\top} + v_{l_{\phi_i}}v_{r_{\phi_i}}^{\top}P_{\phi}) \right) e_{{\phi}}(\tau).
	\end{align*}
    From Lemma \ref{N_E_undirected}, we have that $\mathcal{N}(L_{e_s}) = \mathcal{N}(E_s)$ holds. Then, from the fact that $v_{l_{\phi_1}}^\top L_{e_{s_{\phi}}} = v_{l_{\phi_2}}^\top L_{e_{s_{\phi}}} = \dots = v_{l_{\phi_{\xi}}}^\top L_{e_{s_{\phi}}} = \mathbf{0}$
    we have $v_{l_{\phi_1}}^\top E_{s_{\phi}}^\top = v_{l_{\phi_2}}^\top E_{s_{\phi}}^\top = \dots = v_{l_{\phi_{\xi}}}^\top E_{s_{\phi}}^\top = \mathbf{0}.$
	Then, from the definition of $e_{{\phi}}$ in \eqref{err_vect}, we have that $ v_{l_{\phi_i}}^{\top} e_{{\phi}}(\tau)= 0$. Thus,
    \begin{align*}
		\dot V_{\phi}(e_{{\phi}}(\tau)) &= - \frac{1}{2} k_1 e_{{\phi}}(\tau)^\top e_{{\phi}}(\tau),
	\end{align*}
    in both cases.
    From the definition of $V_{\phi}(e_{{\phi}})$ and Rayleigh theorem \cite[Theorem 4.2.2]{Horn}, we have $\frac{1}{2} \lambda_{\min}(P_{\phi}) \lvert e_{{\phi}}(\tau) \rvert^2 \leq V_{\phi}(e_{{\phi}}(\tau)) \leq \frac{1}{2} \lambda_{\max}(P_{\phi}) \lvert e_{{\phi}}(\tau)\rvert^2,$
	which gives
	\[
	\lvert e_{{\phi}}(\tau)\rvert^2 \geq \frac{2 V_{\phi}(e_{{\phi}}(\tau))}{\lambda_{\max}(P_{\phi})}.
	\]
	Consequently,
	\begin{align}\label{cond1}
		\dot V_{\phi}(e_{{\phi}}(\tau)) \leq -\gamma_{\phi} V_{\phi}(e_{{\phi}}(\tau)),
	\end{align}
	where $\gamma_{\phi} = k_1 \frac{1}{\lambda_{\max}(P_{\phi})}$.
	
	Now, let $\phi, \hat \phi \in \mathcal{P}$ be two consecutive modes and $V_{\phi}(e_{{\phi}}(t_l^+)) := \frac{1}{2} e_{{\phi}}(t_l^+)^\top P_{\phi} e_{{\phi}}(t_l^+)$, $V_{\hat \phi}(e_{{\hat \phi}}(t_l^-)) := \frac{1}{2} e_{{\hat \phi}}(t_l^-)^\top P_{\hat \phi} e_{{\hat \phi}}(t_l^-).$
    It follows from \eqref{edge_switch} that for any $t_l$, 
	\begin{align*}
		\lvert e_{{\phi}}(t_l^+) \rvert^2 &\leq 2\vert \Xi_{\phi, \hat \phi} \rvert^2 \lvert e_{{\hat \phi}}(t_l^-) \rvert^2 + 2\lvert \Phi_l\rvert^2.
	\end{align*}
	Noting that $\vert \Xi_{\phi, \hat \phi} \rvert \equiv 1$, because it is a submatrix of an identity matrix, and the singular values of an identity matrix are all 1, meaning that the largest singular value is also 1,
	\begin{align*}
		V_{\phi}(e_{{\phi}}(t_l^+)) \leq \frac{1}{2}\lvert P_{\phi} \rvert \lvert e_{{\phi}}(t_l^+) \rvert^2 &\leq \lvert P_{\phi} \rvert \lvert e_{{\hat \phi}}(t_l^-) \rvert^2 + \lvert P_{\phi} \rvert \lvert \Phi_l\rvert^2.
	\end{align*}
	Since
    $$ \lvert e_{{\phi}}(t_l^-)\rvert^2 \leq \frac{2 V_{\phi}(e_{{\phi}}(t_l^-))}{\lambda_{\min}(P_{\hat \phi})} $$
    we have
	\begin{align*} 
		\lvert P_{\phi} \rvert \lvert e_{{\hat \phi}}(t_l^-) \rvert^2 + \lvert P_{\phi} \rvert \lvert \Phi_l\rvert^2 \leq 2\lvert P_{\phi} \rvert \frac{V_{\phi}(e_{{\phi}}(t_l^-))}{\lambda_{\min}(P_{\hat \phi})} + \lvert P_{\phi} \rvert \lvert \Phi_l\rvert^2& \\ 
		\leq 2 \frac{\lambda_{\max}(P_{\phi})}{\lambda_{\min}(P_{\hat \phi})} V_{\hat \phi}(e_{x_{\hat \phi}}(t_l^-)) + \lambda_{\max}(P_{\phi}) \lvert \Phi_l\rvert^2& 
	\end{align*}
	which gives  
	\begin{align}\label{cond2}
		V_{\phi}(e_{x_{\phi}}(t_l^+)) \leq \Omega_{\phi , \hat \phi} V_{\hat \phi}(e_{x_{\hat \phi}}(t_l^-)) + \Theta,
	\end{align}
	where $\Omega_{\phi , \hat \phi} = 2 \frac{\lambda_{\max}(P_{\phi})}{\lambda_{\min}(P_{\hat \phi})}$ and $\Theta = \lambda_{\max}(P_{\phi})\lvert \Phi_l\rvert^2$.
	
	From \eqref{cond1}, \eqref{cond2} and invoking Theorem \ref{thmxue}, it follows that the origin of system \eqref{dyn_edge} is globally asymptotically practically stable if the switching $\sigma$ admits an average dwell-time that satisfies \eqref{cond}.
\end{proof}
\begin{thm}\label{prop:result2}
	Consider the OMAS \eqref{FO}, under Assumptions \ref{ass0} and \ref{ass1}, in closed loop with the switching control law \eqref{CL_vect_edge}. Then, the origin of the closed-loop system \eqref{dyn_edge} is globally asymptotically stable, that is, 
	\begin{align}\label{cond_e}
		\lim_{t \to \infty} e_{{\phi}}(t) = 0,
	\end{align}
	if the average dwell time satisfies \eqref{cond}. Furthermore,  let $\mathcal{G}_{s_{\bar \phi}}$ be the signed graph in the last switching mode $\bar \phi = [t_f^+ ,\infty)$.
    \renewcommand{\theenumi}{(\roman{enumi}}
	\begin{enumerate}
    \item If $\mathcal{G}_{s_{\bar \phi}}$ is SB, then agents achieve bipartite consensus, that is, the inequality \eqref{obj_BC} holds. 
    \item If $\mathcal{G}_{s_{\bar \phi}}$ is SUB, then agents achieve trivial consensus, that is, the inequality \eqref{obj_C} holds.
    \end{enumerate}
\end{thm}

\begin{proof}
	After the last switching instant, no new nodes or edges are added to the system nor the signs of the interconnections are changed, then \eqref{edge_switch} holds with $\Xi_{\phi, \hat \phi} = I$ and $\phi_l = 0$. Thus, \eqref{cond2} holds with $\Theta=0$, and from Theorem \ref{thmxue}, it follows that the origin of the system \eqref{FO} is globally asymptotically stable if the switching $\sigma$ admits an average dwell-time that satisfies \eqref{cond}, that is, \eqref{cond_e} holds.
	
	Let $\bar \phi$ be the mode and $\bar \tau \in [t_{f^+}, \infty)$ be the time interval after the last switch. Then, from \eqref{dyn_edge}, we have
	\begin{align*}
		e_{{\bar \phi}}(\bar \tau) = &\left[ e^{-k_1 \lambda_{\bar \phi_1}(L_{e_{s_{\bar \phi}}})\bar \tau}v_{r_{\bar \phi_1}}v_{l_{\bar \phi_1}}^{\top} + \cdots \right.\\
        & \left.+ e^{-k_1 \lambda_{\bar \phi_1}(L_{e_{s_{\bar \phi}}})\bar \tau}v_{r_{\bar \phi_{M_{\bar \phi}}}}v_{l_{\bar \phi_{M_{\bar \phi}}}}^{\top}\right] e_{{\bar \phi}}(t_{f^+}),
	\end{align*}
	where $v_{r_{\bar \phi_i}}$ and $v_{l_{\bar \phi_i}}$ are the right and left eigenvectors associated with the $\xi_{\bar \phi}$ zero eigenvalues of $L_{e_{s_{\bar \phi}}}$, and $\xi_{\bar \phi} = M_{\bar \phi} - N_{\bar \phi} + 1$ in the case that the underlying graph is SB, and $\xi_{\bar \phi} = M_{\bar \phi} - N_{\bar \phi}$, otherwise. Then, 
	\begin{align*}
		\lim_{\bar \tau \to \infty} e_{{\bar \phi}}(\bar \tau) = \left[ v_{r_{\bar \phi_1}}v_{l_{\bar \phi_1}}^{\top} + \cdots + v_{r_{\bar \phi_{\xi_{\bar \phi}}}}v_{l_{\bar \phi_{\xi_{\bar \phi}}}}^{\top}\right] e_{{\bar \phi}}(t_{f^+})&\\
		= \left[ v_{r_{\bar \phi_1}}v_{l_{\bar \phi_1}}^{\top} + \cdots + v_{r_{\bar \phi_{\xi_{\bar \phi}}}}v_{l_{\bar \phi_{\xi_{\bar \phi}}}}^{\top}\right] E_{s_{\bar \phi}}^\top {x}(t_{f^+}) = 0.&
	\end{align*}
	Since $v_{l_{\bar \phi_1}}^{\top} E_{s_{\bar \phi}}^\top = v_{l_{\bar \phi_2}}^{\top} E_{s_{\bar \phi}}^\top = \cdots = v_{l_{\bar \phi_{\xi_{\bar \phi}}}}^{\top} E_{s_{\bar \phi}}^\top = 0$, it leads to $\lim_{\bar \tau \to \infty} E_{s_{\bar \phi}}^\top {x}(\bar \tau)= 0$. 
    
   (i) In the case the signed graph is SB, using the gauge transformation, we have that 
    $$\lim_{\bar \tau \to \infty} E_{s_{\bar \phi}}^\top {x}(\bar \tau)= \lim_{\bar \tau \to \infty} D_{e_{\bar \phi}} E_{{\bar \phi}}^\top D_{{\bar \phi}} {x}(\bar \tau) = 0.$$ 
    From \cite[Theorem 3.4]{zelazo2007agreement}, the null space of $E_{{\phi}}^\top$ is spanned by $\mathbf{1}_{N_{\phi}}$, so $\lim_{\bar \tau \to \infty} D_{{\bar \phi}} {x}(\bar \tau) = \alpha \mathbf{1}_{N_{\bar \phi}},$ where $\alpha \in \mathbb{R}$. Therefore, we can deduce that
	\begin{align*}
		\lim_{\bar \tau \to \infty} {x}(\bar \tau) = \alpha D_{{\bar \phi}} \mathbf{1}_{N_{\bar \phi}},
	\end{align*}
	which implies that the system achieves bipartite consensus.
	
	(ii) In the case that the signed graph is SUB, we have that $\lim_{\bar \tau \to \infty} E_{s_{\bar \phi}}^\top {x}(\bar \tau) = 0,$ and since from Lemma \ref{lemma2}, $\mbox{rank}(E_{s_{\bar \phi}}) = N_{\bar \phi}$, the only solution is 
	\begin{align*}
		\lim_{\bar \tau \to \infty}{x}(\bar \tau) = 0,
	\end{align*}
	which implies that the system achieves trivial consensus.
\end{proof}

\section{Simulation Results}\label{section6}
To illustrate our theoretical findings, we present a numerical example involving a system of multi-wheeled mobile robots modeled as unicycles \cite{tzafestas2013_book}. Let $\begin{bmatrix}r_{x_i} & r_{y_i} \end{bmatrix}^{\top} \in \mathbb{R}^2$ be the position of the center of the $i$th robot, $\theta_i \in \mathbb{R}$ its orientation, and $v_i \in \mathbb{R}$ and $\omega_i \in \mathbb{R}$ its linear and angular velocities, respectively. The dynamics of the wheeled mobile robots can then be described as
\begin{align} \label{480}
	\dot{r}_{x_i} = v_i\cos(\theta_i),\quad \dot{r}_{y_i}=v_i\sin(\theta_i),\quad \dot{\theta}_i = \omega_i.
\end{align}
To implement the control law \eqref{CL_vect_edge} on this system, we first apply a preliminary feedback linearizing control and we express the system's dynamics in terms of the position of a point located at a distance $\delta$ from the axis connecting the wheels. That is, we define the point $p_i = \begin{bmatrix}p_{x_i} & p_{y_i} \end{bmatrix}^{\top}$, where $p_{x_i} = r_{x_i} + \delta_i\cos(\theta_i)$ and $p_{y_i} = r_{y_i} + \delta_i\sin(\theta_i)$. For simulation purposes, we set $\delta_i = 0.1$m. Differentiating $p_i$ with respect to time and by letting
\begin{align} \label{457}
	\begin{bmatrix} v_i \\ \omega_i \end{bmatrix} = \begin{bmatrix} \cos(\theta_i) & \sin(\theta_i) \\ -\frac{1}{\delta_i}\sin(\theta_i) & \frac{1}{\delta_i}\cos(\theta_i) \end{bmatrix}\begin{bmatrix} u_{x_i} \\ u_{y_i} \end{bmatrix},
\end{align}
we get $\begin{bmatrix} \dot{p}_{x_i} & \dot{p}_{y_i} \end{bmatrix}^{\top} = \begin{bmatrix} u_{x_i} & u_{y_i} \end{bmatrix}^{\top}$,
which is a simplified kinematic equation in the form of first-order dynamics. 
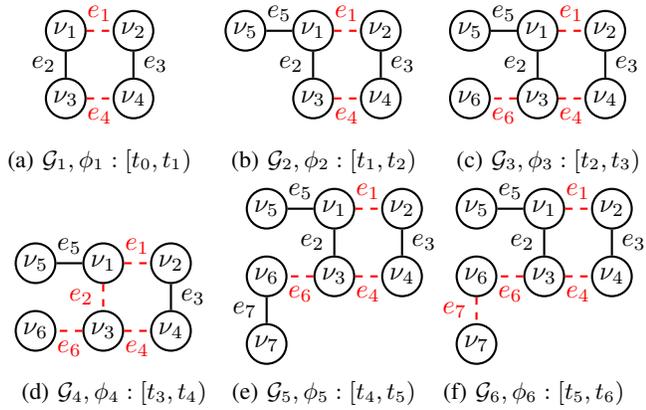
\begin{figure}[t]
	\centering
	\begin{subfigure}[b]{0.15\textwidth}
		\centering
\begin{tikzpicture}[node distance={9mm}, thick,main/.style = {draw, circle}] 
	\tikzset{mynode/.style={circle,draw,minimum size=15pt,inner sep=0pt,thick},}
	\node[mynode] (1) {$\nu_1$}; 
	\node[mynode] (2) [right  of=1] {$\nu_2$};
	\node[mynode] (3) [below of=1] {$\nu_3$};
	\node[mynode] (4) [below of=2] {$\nu_4$};
    \draw[red, dash pattern=on 1mm off 1mm][-](1) -- node[midway, above] {$e_1$}(2);
	\draw[-] (1) -- node[midway, left] {$e_2$}(3);
	\draw[-] (2) -- node[midway, right] {$e_3$}(4);
	\draw[-, color=red, dash pattern=on 1mm off 1mm] (3) -- node[midway, below] {$e_4$}(4);
	\end{tikzpicture}
\caption{$\mathcal{G}_1, \phi_1:[t_0, t_1)$}
     \end{subfigure}
\hfill
\begin{subfigure}[b]{0.15\textwidth}
			\centering
	\begin{tikzpicture}[node distance={9mm}, thick,main/.style = {draw, circle}] 
		\tikzset{mynode/.style={circle,draw,minimum size=15pt,inner sep=0pt,thick},}
		\node[mynode] (1) {$\nu_1$}; 
		\node[mynode] (2) [right  of=1] {$\nu_2$};
		\node[mynode] (3) [below of=1] {$\nu_3$};
		\node[mynode] (4) [below of=2] {$\nu_4$};
		\node[mynode] (5) [left of=1] {$\nu_5$};
    \draw[red, dash pattern=on 1mm off 1mm][-](1) -- node[midway, above] {$e_1$}(2);
	\draw[-] (1) -- node[midway, left] {$e_2$}(3);
	\draw[-] (2) -- node[midway, right] {$e_3$}(4);
	\draw[-, color=red, dash pattern=on 1mm off 1mm] (3) -- node[midway, below] {$e_4$}(4);
		\draw[-] (1) -- node[midway, above] {$e_5$}(5);
	\end{tikzpicture}
\caption{$\mathcal{G}_2, \phi_2:[t_1, t_2)$}
\end{subfigure}
\hfill
\begin{subfigure}[b]{0.15\textwidth}
	\centering
	\begin{tikzpicture}[node distance={9mm}, thick,main/.style = {draw, circle}] 
		\tikzset{mynode/.style={circle,draw,minimum size=15pt,inner sep=0pt,thick},}
		\node[mynode] (1) {$\nu_1$}; 
		\node[mynode] (2) [right  of=1] {$\nu_2$};
		\node[mynode] (3) [below of=1] {$\nu_3$};
		\node[mynode] (4) [below of=2] {$\nu_4$};
		\node[mynode] (5) [left of=1] {$\nu_5$};
		\node[mynode] (6) [left of=3] {$\nu_6$};
    \draw[red, dash pattern=on 1mm off 1mm][-](1) -- node[midway, above] {$e_1$}(2);
	\draw[-] (1) -- node[midway, left] {$e_2$}(3);
	\draw[-] (2) -- node[midway, right] {$e_3$}(4);
	\draw[-, color=red, dash pattern=on 1mm off 1mm] (3) -- node[midway, below] {$e_4$}(4);
		\draw[-] (1) -- node[midway, above] {$e_5$}(5);
		\draw[-, color=red, dash pattern=on 1mm off 1mm] (3) -- node[midway, below] {$e_6$}(6);
	\end{tikzpicture}
\caption{$\mathcal{G}_3, \phi_3:[t_2, t_3)$}
\end{subfigure}
\begin{subfigure}[b]{0.15\textwidth}
	\centering
	\begin{tikzpicture}[node distance={9mm}, thick,main/.style = {draw, circle}] 
		\tikzset{mynode/.style={circle,draw,minimum size=15pt,inner sep=0pt,thick},}
		\node[mynode] (1) {$\nu_1$}; 
		\node[mynode] (2) [right  of=1] {$\nu_2$};
		\node[mynode] (3) [below of=1] {$\nu_3$};
		\node[mynode] (4) [below of=2] {$\nu_4$};
		\node[mynode] (5) [left of=1] {$\nu_5$};
		\node[mynode] (6) [left of=3] {$\nu_6$};
    \draw[red, dash pattern=on 1mm off 1mm][-](1) -- node[midway, above] {$e_1$}(2);
		\draw[-, color=red, dash pattern=on 1mm off 1mm] (1) -- node[midway, left] {$e_2$}(3);
		\draw[-] (2) -- node[midway, right] {$e_3$}(4);
	\draw[-, color=red, dash pattern=on 1mm off 1mm] (3) -- node[midway, below] {$e_4$}(4);
		\draw[-] (1) -- node[midway, above] {$e_5$}(5);
		\draw[-, color=red, dash pattern=on 1mm off 1mm] (3) -- node[midway, below] {$e_6$}(6);
	\end{tikzpicture}
	\caption{$\mathcal{G}_4, \phi_4:[t_3, t_4)$}
\end{subfigure}
\begin{subfigure}[b]{0.15\textwidth}
	\centering
	\begin{tikzpicture}[node distance={9mm}, thick,main/.style = {draw, circle}] 
		\tikzset{mynode/.style={circle,draw,minimum size=15pt,inner sep=0pt,thick},}
		\node[mynode] (1) {$\nu_1$}; 
		\node[mynode] (2) [right  of=1] {$\nu_2$};
		\node[mynode] (3) [below of=1] {$\nu_3$};
		\node[mynode] (4) [below of=2] {$\nu_4$};
		\node[mynode] (5) [left of=1] {$\nu_5$};
		\node[mynode] (6) [left of=3] {$\nu_6$};
		\node[mynode] (7) [below of=6] {$\nu_7$};
    \draw[red, dash pattern=on 1mm off 1mm][-](1) -- node[midway, above] {$e_1$}(2);
	\draw[-] (1) -- node[midway, left] {$e_2$}(3);
	\draw[-] (2) -- node[midway, right] {$e_3$}(4);
	\draw[-, color=red, dash pattern=on 1mm off 1mm] (3) -- node[midway, below] {$e_4$}(4);
		\draw[-] (1) -- node[midway, above] {$e_5$}(5);
		\draw[-, color=red, dash pattern=on 1mm off 1mm] (3) -- node[midway, below] {$e_6$}(6);
		\draw[-] (6) -- node[midway, left] {$e_7$}(7);
	\end{tikzpicture}
	\caption{$\mathcal{G}_5, \phi_5:[t_4, t_5)$}
\end{subfigure}
\begin{subfigure}[b]{0.15\textwidth}
	\centering
	\begin{tikzpicture}[node distance={9mm}, thick,main/.style = {draw, circle}] 
		\tikzset{mynode/.style={circle,draw,minimum size=15pt,inner sep=0pt,thick},}
		\node[mynode] (1) {$\nu_1$}; 
		\node[mynode] (2) [right  of=1] {$\nu_2$};
		\node[mynode] (3) [below of=1] {$\nu_3$};
		\node[mynode] (4) [below of=2] {$\nu_4$};
		\node[mynode] (5) [left of=1] {$\nu_5$};
		\node[mynode] (6) [left of=3] {$\nu_6$};
		\node[mynode] (7) [below of=6] {$\nu_7$};
    \draw[red, dash pattern=on 1mm off 1mm][-](1) -- node[midway, above] {$e_1$}(2);
	\draw[-] (1) -- node[midway, left] {$e_2$}(3);
	\draw[-] (2) -- node[midway, right] {$e_3$}(4);
	\draw[-, color=red, dash pattern=on 1mm off 1mm] (3) -- node[midway, below] {$e_4$}(4);
		\draw[-] (1) -- node[midway, above] {$e_5$}(5);
		\draw[-, color=red, dash pattern=on 1mm off 1mm] (3) -- node[midway, below] {$e_6$}(6);
		\draw[-, color=red, dash pattern=on 1mm off 1mm] (6) -- node[midway, left] {$e_7$}(7);
	\end{tikzpicture}
	\caption{$\mathcal{G}_6, \phi_6:[t_5, t_6)$}
\end{subfigure}
  \caption{ Black lines represent cooperative interactions, and dashed red lines represent antagonistic interactions. (a) The initial graph $\mathcal{G}_1$ is a SB signed network of $4$ agents, where $\mathcal{V}_1 = \{\nu_1, \nu_3\},\ \mathcal{V}_2 = \{\nu_2, \nu_4\}$. (b) At $t=t_1$, a new node $\nu_5$ is added to the system and $\mathcal{G}_2$ is a SB signed network of $5$ agents, where $\mathcal{V}_1 = \{\nu_1, \nu_3, \nu_5\}, \mathcal{V}_2 = \{\nu_2,\nu_4\}$. (c) At $t=t_2$, a new node $\nu_6$ is added to the system and $\mathcal{G}_3$ is a SB signed network of $6$ agents, where $\mathcal{V}_1 = \{\nu_1, \nu_3, \nu_5\}, \mathcal{V}_2 = \{\nu_2, \nu_4, \nu_6\}$. (d) At $t=t_3$, the sign of the edge $e_2$ changes from cooperation to antagonism, and $\mathcal{G}_4$ is a SUB signed network. (e) At $t=t_4$, the sign of the edge $e_2$ changes back to cooperation and a new node $\nu_7$ is added to the system. $\mathcal{G}_5$ is a SB signed network of $7$ agents, where $\mathcal{V}_1 = \{\nu_1, \nu_3, \nu_5 \}, \mathcal{V}_2 = \{\nu_2, \nu_4, \nu_6 , \nu_7\}$. (f) At $t=t_5$, the sign of the edge $e_7$ changes from cooperation to antagonism, and $\mathcal{G}_5$ is a SB signed network of $7$ agents, where $\mathcal{V}_1 = \{\nu_1, \nu_3, \nu_5, \nu_7 \}, \mathcal{V}_2 = \{\nu_2, \nu_4, \nu_6 \}$.}
  \label{graph}
\end{figure}

We consider a network initially represented by a signed, undirected graph with four agents. Over time, new agents are added, and the signs of interconnections change, as depicted in Figure \ref{graph}. The initial states of the agents are provided in Table \ref{table1}, and the switching times are listed in Table \ref{table2}.

\begin{table}[h]
	\centering
    \small
    \setlength{\tabcolsep}{4pt} 
    \renewcommand{\arraystretch}{0.9}
    \caption{ Initial conditions of the agents}\label{table1}
		\begin{tabular}{c|ccccccccc}
			\toprule
			Agents & $\nu_1$ & $\nu_2$ & $\nu_3$ & $\nu_4$ & $\nu_5$ & $\nu_6$ & $\nu_7$ \\
			\midrule
			$r_x\ [m]$ & 3.5 & 4 & -2 & -6.5 & 5.5 & -10.5 & 3.5 \\
			$r_y\ [m]$ & 2 & 3.5 & -3 & -1 & -3 & 6 & 3.5 \\
			\bottomrule
		\end{tabular}
		\label{tab:init_conditions}
\end{table} 

\begin{table}[h]
\caption{Switching times \( t = t_i \) [s]}\label{table2}
	\centering
       \small
    \setlength{\tabcolsep}{4pt} 
    \renewcommand{\arraystretch}{0.9}
		\begin{tabular}{cccccc}
			\toprule 
			$t_1$ & $t_2$ & $t_3$ & $t_4$ & $t_5$\\
			\midrule
			1.3 & 2.5 & 5.5 & 10 & 15 \\
			\bottomrule
		\end{tabular}
\end{table}

\begin{figure}[h!]
	\centering
	\includegraphics[width=\columnwidth]{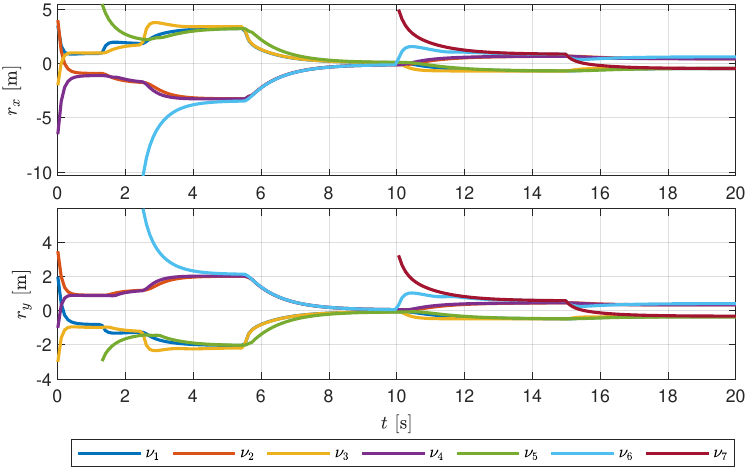}
	\caption{ {\small Evolution of the trajectories of the agents' positions.}}
    \label{fig2}
\end{figure}

\begin{figure}[h!]
	\centering
	\includegraphics[width=\columnwidth]{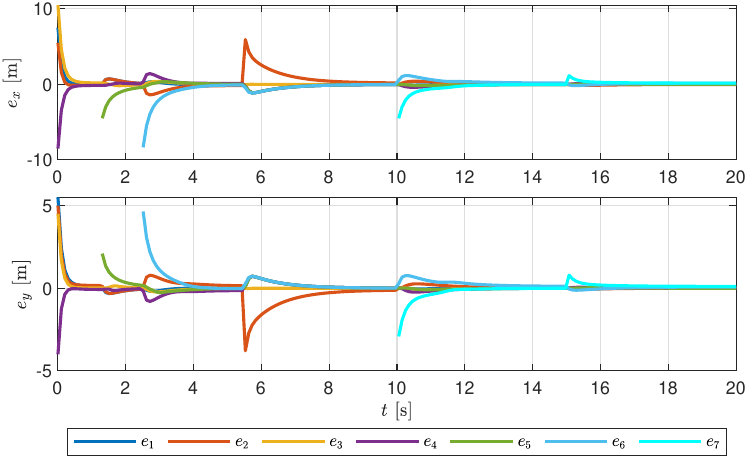}
	\caption{{\small Evolution of the trajectories of the edges.}}
    \label{fig3}
\end{figure}
The evolution of the agents' position trajectories is shown in Figure \ref{fig2}, while the synchronization errors (edges states) are depicted in Figure \ref{fig3}. The initial graph $\mathcal{G}_1$ is SB, leading the agents to converge around two symmetric equilibria. At $t=t_1$, a new node is added, and since $\mathcal{G}_2$ remains SB, the agents continue to converge around two equilibria while accounting for the newly added agent. Similarly, at $t=t_2$ and $t=t_4$, new nodes are added to the system, and since $\mathcal{G}_3$ and $\mathcal{G}_5$ are SB, the agents maintain convergence to two equilibria. However, at $t=t_3$, a sign change in the interconnection between nodes $\nu_1$ and $\nu_3$ renders $\mathcal{G}_4$ SUB, causing the agents to converge to zero. Finally, after the last switch at $t=t_6$, the graph $\mathcal{G}_6$ is SB, allowing the agents to achieve bipartite consensus. The peaks on Figure \ref{fig3} represent the addition of new edges or sign changes in the interconnections. 

\section{Conclusions}\label{section7}
This paper investigates the problem of bipartite consensus in networks with cooperative and antagonistic interactions, where new nodes and edges can be introduced, and the nature of interactions may change over time. Using a Lyapunov-based approach, we analyze the global asymptotic stability of bipartite consensus control in open multi-agent systems over undirected signed graphs. As future work, we aim to extend these results to general signed digraphs that exhibit both SB and SUB properties and to explore the impact of multiple leaders in the network. Specifically, we are currently working on a paper that addresses the behavior of OMAS over general signed digraphs, which will include a broader analysis of directed interactions and their stability.

\bibliographystyle{IEEEtran}
\bibliography{references_pelin,references_angela}

\end{document}